\documentclass[%
reprint,
nofootinbib,
amsmath,amssymb,
aps,
]{revtex4-2}

\usepackage{bm}
\usepackage{braket}
\usepackage{amsthm}
\usepackage{mathtools}
\usepackage{amssymb}
\usepackage{amsmath}
\usepackage{amssymb}
\usepackage{esvect}
\usepackage{dsfont}
\usepackage{tensor}
\usepackage{slashed}
\usepackage{float}
\usepackage{graphicx}
\usepackage{enumitem}
\usepackage{xcolor}
\usepackage{mathrsfs} 

\usepackage[caption=false]{subfig} 

\newtheorem{theorem}{Theorem} 

\newtheorem{thm}[theorem]{Theorem}

\newtheorem{defn}[theorem]{Definition}

\newtheorem{prop}[theorem]{Proposition}

\usepackage{graphicx}
\usepackage{dcolumn}
\usepackage{bm}
\usepackage{hyperref}
\hypersetup{
colorlinks = true,
citecolor= blue,
linkcolor= black
}

\newcommand{\ketbra}[2]{\ket{#1}\!\bra{#2}}
\newcommand{\vect}[1]{\mathbf{#1}}
\newcommand{\ptope}{\mathscr{P}}
\newcommand{\subn}[1]{\breve{\vect{#1}}}

\newcommand{\overlap}{g} 

\newcommand{\markup}[1]{#1}

\begin{document}

\title{Fidelity Bounds for Device-Independent Advantage Distillation}

\author{Thomas A. Hahn}
\email{thomas.hahn@weizmann.ac.il}
\affiliation{Institute for Theoretical Physics, ETH Z{\"u}rich, Switzerland}
\author{Ernest Y.-Z. Tan}
\email{yzetan@uwaterloo.ca}
\affiliation{Institute for Theoretical Physics, ETH Z{\"u}rich, Switzerland}

\begin{abstract}
It is known that advantage distillation (that is, information reconciliation using two-way communication) improves noise tolerances for quantum key distribution (QKD) setups. Two-way communication is hence also of interest in the device-independent case, where noise tolerance bounds for one-way error correction are currently too low to be experimentally feasible. Existing security proofs for the device-independent repetition-code protocol (the most prominent form of advantage distillation) rely on fidelity-related security conditions, but previous bounds on the fidelity were not tight. We improve on those results by developing an algorithm that returns arbitrarily tight lower bounds on the fidelity. Our results give insight on how strong the fidelity-related security conditions are, and could also be used to compute some lower bounds on one-way protocol keyrates. Finally, we conjecture a necessary security condition for the protocol studied in this work, that naturally complements the existing sufficient conditions. 
\end{abstract}

\maketitle

\section{Introduction} \label{sec:protocol}
The ultimate goal of key distribution protocols is to generate secure keys between two parties, Alice and Bob. To this end, device-independent quantum key distribution (DIQKD) schemes aim to provide information-theoretically secure keys by taking advantage of non-local correlations, which can be verified via Bell inequalities \cite{pironio2009deviceindependent, scarani2013deviceindependent,Arnon-Friedman:2018aa,bennett2020quantum}. Critically, Bell violations rely only on the measurement statistics, $\operatorname{Pr}\left(ab|xy\right)$, where $a \ (b)$ is Alice's (Bob's) measurement outcome and $x \ (y) $ is Alice's (Bob's) measurement setting \cite{brunner2013bell}. By basing security on Bell inequalities, DIQKD protocols do not require any knowledge of the bipartite state that Alice and Bob share, nor of the measurements both parties conduct, apart from the assumption that they act on separate Hilbert spaces \cite{brunner2013bell,pironio2009deviceindependent}. (To guarantee security, it is still important to ensure that information about the device outputs themselves is not simply leaked to the adversary. Also, if the devices are reused, they must not access any registers retaining memory of `private data', in order to avoid the memory attack of~\cite{PhysRevLett.110.010503}.) We consider all Hilbert spaces to be finite-dimensional.

Although DIQKD allows for the creation of secret keys under very weak assumptions, there is a trade-off when it comes to noise tolerances \cite{1176619,pironio2009deviceindependent}. Standard DIQKD protocols which apply one-way error-correction steps have fairly low noise robustness, and are therefore not currently experimentally feasible \cite{pironio2009deviceindependent, Arnon-Friedman:2018aa}. To improve noise tolerances, one may implement techniques such as noisy pre-preprocessing \cite{PhysRevLett.124.230502}, basing the protocol on asymmetric CHSH inequalities \cite{woodhead2020deviceindependent,sekatski2020deviceindependent}, or applying advantage distillation \cite{tan2019advantage}. In this work we focus solely on advantage distillation, which refers to using two-way communication for information reconciliation, in place of one-way error-correction. 

In device-dependent QKD \cite{1176619,PhysRevA.66.060302,renner2005security,bae2006key,PhysRevA.73.012327, Watanabe_2007,Khatri2017NumericalEF}, as well as classical key distillation scenarios \cite{256484,Wolf99}, advantage distillation can perform better than standard one-way protocols 
in terms of noise tolerance. As for DIQKD, it has been shown that advantage distillation leads to an improvement of noise tolerance as well \cite{tan2019advantage}, but the results obtained in that work may not be optimal. Specifically, a sufficient condition was derived there for the security of advantage distillation against collective attacks, 
based on the fidelity between an appropriate pair of conditional states (see Theorem~\ref{Theoremsufffid1} below). However, the approach used in~\cite{tan2019advantage} to bound this fidelity is suboptimal, and hence the results were not tight.

Our main contribution in this work is to derive an algorithm based on semidefinite programs (SDPs) that yields arbitrarily tight lower bounds on the relevant fidelity quantity considered in~\cite{tan2019advantage}. We apply this algorithm to several DIQKD scenarios studied in that work, and compare the resulting bounds. 
Surprisingly, while we find improved noise tolerance for some scenarios, we do not have such improvements for the scenario that gave the best noise tolerances in~\cite{tan2019advantage}, which relied on a more specialized security argument.
An important consequence of this finding is that it serves as strong evidence that the general sufficient condition described in~\cite{tan2019advantage} is in fact not necessary in the DIQKD setting, in stark contrast to the device-dependent QKD protocols in~\cite{bae2006key,PhysRevA.73.012327}, where it is both necessary and sufficient (when focusing on the repetition-code advantage distillation protocol; see below). In light of this fact, we describe an analogous condition that we conjecture to be necessary, and discuss possible directions for further progress.

We consider the following set-up for two parties, Alice and Bob \cite{scarani2013deviceindependent,tan2019advantage}.
\begin{itemize}
	\item Measurement Settings: Alice (Bob) has $M_A$ ($M_B$) possible measurement inputs and chooses $x\in \mathcal{X} = \{ 0,..., M_A-1 \}$ $\left(y \in \mathcal{Y} = \{ 0,..., M_B-1 \} \right)$. The measurements are denoted as $A_{x}$ ($B_{y}$) for Alice (Bob). 
	\item Measurement Outcomes: Alice (Bob) has $2$ possible measurement outcomes and measures ${a \in \mathcal{A} = \{ 0,1 \}}$ $\left(b \in \mathcal{B} = \{ 0,1 \}\right)$.
\end{itemize}
The only key-generating measurements are $A_0$ and $B_0$. We consider Eve to be restricted to collective attacks \cite{Scarani_2009}, where it is assumed that the measurements Alice and Bob may conduct, as well as the single-round tripartite state, $\rho_{ABE}$, that Alice, Bob, and an adversary, Eve, share are independent and identical for each round. 
Since we are working in the device-independent setting, we consider Alice and Bob's measurements to be otherwise uncharacterized.
For ease of applying the results from~\cite{tan2019advantage}, we assume that they perform a symmetrization step, in which Alice and Bob publicly communicate a uniformly random bit and XOR it with their raw outputs (see~\cite{tan2019advantage} for details on when this step can be omitted). We use $\epsilon$ to denote the quantum bit error rate (QBER), i.e.~the probability of obtaining different outcomes from measurements $A_0, B_0$.

We consider the repetition-code protocol for advantage distillation \cite{256484,Wolf99,renner2005security,bae2006key,PhysRevA.73.012327}, which proceeds as follows. After gathering the raw output strings from their devices, Alice and Bob split the key-generating rounds into blocks of size $n$ each, and from each block they will attempt to generate a single highly-correlated bit. For each block, Alice generates a secret, uniformly random bit, $C$, and adds it to her $n$ bits, $\textbf{A}_0$. She then sends this `encoded' bitstring, $\textbf{M} = \textbf{A}_0 \oplus (C,...,C)$, as a message to Bob via an authenticated channel. He then tries to decode the message by adding his own bitstring, $\textbf{B}_0$, to it. He accepts the block of $n$ bits if and only if $\textbf{M} \oplus \textbf{B}_0 = (C',...,C')$. If accepted, he indicates this to Alice by sending her a bit $D=1$ via an authenticated channel. Otherwise, he sends $D=0$. Considering only the accepted blocks, this process has therefore reduced each block of $n$ bitpairs to a single highly correlated bitpair, $(C,C')$. 
Alice and Bob then apply a one-way error correction procedure (from Alice to Bob) on the resulting bitpairs over asymptotically many rounds, followed by privacy amplification to produce a final secret key. This procedure can achieve a positive asymptotic keyrate if the bitpairs in the accepted blocks satisfy some conditions we shall now describe. 

The protocol can be used to distill a secret key if \cite{Devetak_2005} 
\begin{equation} \label{eq:entropycond1}
r \coloneqq \operatorname{H}(C|\mathbf{E} \mathbf{M}; D=1) - \operatorname{H}(C|C'; D=1) > 0 \ ,
\end{equation}
where $\mathbf{E}$ is Eve's side-information across one block of $n$ rounds and $\operatorname{H}$ is the von Neumann entropy  \cite{Arnon-Friedman:2018aa,Devetak_2005}. The second entropy term, $\operatorname{H}(C|C'; D=1)$, can easily be determined via the QBER \cite{tan2019advantage,pironio2009deviceindependent}. In \cite{tan2019advantage}, Eve's conditional entropy $\operatorname{H}(C|\mathbf{E} \mathbf{M}; D=1)$ was lower-bounded using inequalities that are not necessarily tight, leading to the following result based on the (root-)fidelity $\operatorname{F}(\rho, \sigma) \coloneqq \lVert \sqrt{\rho} \sqrt{\sigma}\rVert_1$:
\begin{thm} \label{Theoremsufffid1}
A secret key can be generated if 
\begin{equation} \label{eq:sufffidcond1}
\operatorname{F}(\rho_{E|00}, \rho_{E|11})^2 > \frac{\epsilon}{(1-\epsilon)} \ ,
\end{equation}
where $\rho_{E|a_0 b_0}$ denotes Eve's conditional state (in a single round) after Alice and Bob use inputs $A_0$ and $B_0$ and obtain outcomes $a_0$ and $b_0$.
\end{thm} 
\noindent Our goal will be to find a general method to certify the condition in Theorem~\ref{Theoremsufffid1}.
For later use, we also note that in the case where both parties have binary inputs and outputs, an alternative condition was derived~\cite{tan2019advantage} based on the trace distance $d(\rho, \sigma) \coloneqq (1/2) \lVert \rho-\sigma \rVert_1 $:
\begin{thm} \label{Theoremsuffdist}
If $\mathcal{X} = \mathcal{Y}=\{ 0,1 \}$ and all measurements have binary outcomes, a secret key can be generated if 
\begin{equation} \label{eq:sufftracecond1}
1-d( \rho_{E|00}, \rho_{E|11})> \frac{\epsilon}{(1-\epsilon)} \ .
\end{equation}
\end{thm} 

\section{Results}
To find optimal bounds for Theorem \ref{Theoremsufffid1}, we need to minimize the fidelity for a given observed distribution. We show that this can be written as an SDP in Sec.~\ref{sec:SDP}, and use this to calculate noise tolerances for a range of repetition-code advantage distillation setups in Sec.~\ref{sec:thresholds}. We conclude the Results section with a conjecture for a necessary condition that naturally complements the sufficient condition in \cite{tan2019advantage}.

\subsection{SDP Formulation of Minimum Fidelity Condition}\label{sec:SDP}
To see if Eve can minimize the fidelity such that \eqref{eq:sufffidcond1} does not hold, we must first solve the following constrained optimization over all possible $\rho_{ABE}$ and possible measurements by Alice and Bob:
\begin{equation} \label{eq:exactfidopt1}
\begin{aligned}
\inf \quad & \operatorname{F}(\rho_{E|00}, \rho_{E|11}) \\
\textrm{s.t.} \quad &\operatorname{Pr}(ab|xy)_{\rho} = \vect{p} \ ,
\end{aligned}
\end{equation}
where $\operatorname{Pr}(ab|xy)_{\rho}$ denotes the combined outcome probability distribution that would be obtained from $\rho_{ABE}$ (and some measurements), and $\vect{p}$ represents the measurement distribution Alice and Bob actually observe. We observe that after Alice and Bob perform the key-generating measurements, the resulting tripartite state is of the form
\begin{equation} \label{eq:betainitialstate}
\sum_{a,b \in\{0,1\}} \operatorname{Pr}(ab)  \ketbra{ab}{ab} \otimes \rho_{E|ab} \ ,
\end{equation}
where for brevity we use $\operatorname{Pr}(ab)$ to denote the probability of getting outcomes $(a,b)$ when the key-generating measurements are performed, i.e.~$\operatorname{Pr}(ab|00)_{\rho}$.

To turn this optimization into an SDP, we first note that for any pair $(\rho_{E|00}, \rho_{E|11})$, there exists a measurement Eve can perform that leaves the fidelity invariant \cite{fuchs1996mathematical}. 
Also, since this measurement is on Eve's system only, performing it does not change the value $\operatorname{Pr}(ab|xy)_{\rho}$ in the constraint either.
Therefore, given any feasible $\rho_{ABE}$ and measurements in the optimization, we can produce another feasible state and measurements with the same objective value but with Eve having performed the \cite{fuchs1996mathematical} measurement
that leaves the fidelity between the original $(\rho_{E|00}, \rho_{E|11})$ states invariant. 

After performing this measurement, the state~\eqref{eq:betainitialstate} becomes
\begin{equation} \label{eq:Initialstate}
\sum_{a,b\in\{0,1\}} \sum_{i} \operatorname{Pr}(ab| i)  \operatorname{Pr}(i) \ketbra{abi}{abi} \ ,
\end{equation}
where the index $i$ represents the possible outcomes for Eve's measurement (we do not limit the number of such outcomes for now). For these states, the fidelity can be written (assuming the distribution is symmetrized) as
\begin{equation} \label{eq:fidfct1}
\operatorname{F}(\rho_{E|00},\rho_{E|11}) =  \sum_{i} \frac{\sqrt{ \operatorname{Pr}(00| i) \operatorname{Pr}(11| i) }}{\operatorname{Pr}(00)}  \operatorname{Pr}(i) \ .
\end{equation}
As for the constraints, we note that this measurement by Eve commutes with Alice and Bob's measurements, hence we can write $\operatorname{Pr}(ab|xy)_{\rho} = \sum_{i} \operatorname{Pr}(i) \vect{p^{i}}$, where $\vect{p^{i}}$ denotes the Alice-Bob distribution conditioned on Eve getting outcome $i$. Note that $\vect{p^{i}}$ is always a distribution realizable by quantum states and measurements, since conditioning on Eve getting outcome $i$ produces a valid quantum state on Alice and Bob's systems.

The solution to \eqref{eq:exactfidopt1} is therefore equal to the output of the optimization problem 
\begin{equation}  \label{eq:exactfidopt2}
\begin{aligned}
\inf_{\operatorname{Pr}(i),\vect{p^{i}}} \quad &  \sum_{i} \frac{\sqrt{ \operatorname{Pr}(00| i) \operatorname{Pr}(11| i) }}{\operatorname{Pr}(00)}   \operatorname{Pr}(i) \\
\textrm{s.t.} \quad &\sum_{i} \operatorname{Pr}(i) \vect{p}^{{i}} =\vect{p} \\
\quad &\vect{p}^{i} \in  \mathcal{Q}_{\mathcal{X},\mathcal{Y}} \\
\quad &\operatorname{Pr}(i) \in   \mathcal{P(\mathcal{I})} \ ,
\end{aligned}
\end{equation}
where $\mathcal{Q}_{\mathcal{X},\mathcal{Y}}$ represents the set of quantum realizable distributions, and $\mathcal{P(\mathcal{I})}$ is the set of probability distributions $\operatorname{Pr}(i)$ on Eve's (now classical) side-information. We note that the constraints can be relaxed to a convergent hierarchy of SDP conditions, following the approach in \cite{navascues2008convergent,bancal2013randomness,Nieto_Silleras_2014}, but the objective function is not affine. To address this, we show in Sec.~\ref{Sec:Polytopes} that since the objective function is a convex sum of bounded concave functions, it can be approximated arbitrarily well using upper envelopes of convex polytopes. This will allow us to lower-bound this optimization using an SDP hierarchy, and do so without any 
knowledge about the dimension of Eve's system other than the assumption that all Hilbert spaces are finite-dimensional. Critically, SDPs have the important property that they yield certified lower bounds on the minimum value (via the dual value), so we can be certain that we truly have a lower bound on the optimization~\eqref{eq:exactfidopt1}. 
Our approach is based on the SDP reduction in \cite{himbeeck2019correlations}; however, we give a more detailed description and convergence analysis 
for the situation where the optimization involves concave functions of more than one variable (which is required in our work but not necessarily in~\cite{himbeeck2019correlations}).

We remark that in this work, we focus on the situation where the constraints in~\eqref{eq:exactfidopt2} involve the full list of output probabilities. However, our analysis generalizes straightforwardly to situations where the constraints only consist of one or more linear combinations of these probabilities (though it is still necessary to have an estimate of the $\operatorname{Pr}(00)$ term in the denominator of the objective function), which can slightly simplify the corresponding DIQKD protocols. 

\subsection{Results From SDP}\label{sec:thresholds}
With the exception of the 2-input scenario for both parties (where Theorem~\ref{Theoremsuffdist} can be applied instead of Theorem~\ref{Theoremsufffid1}), previous bounds for the fidelity were calculated via the Fuchs--van de Graaf inequalities \cite{tan2019advantage,fuchs1997cryptographic}. To compare this against our method, we consider the depolarizing noise model \cite{nielsen_chuang_2010}, i.e.~where the observed statistics have the form
\begin{equation}
\left(1-2q\right) \operatorname{Pr_{target}}(ab|xy) +q/2 \ ,
\end{equation}
where $\operatorname{Pr_{target}}$ refers to some target distribution in the absence of noise. We consider three possible choices for the target distribution, and show our results for these scenarios in Fig.~\ref{fig:grouped}:
\begin{enumerate}[label=(\alph*)]
	\item Both Alice and Bob have four possible measurement settings. The target distribution is generated by:
	\begin{itemize}
		\item $\ket{\phi^+}=\left(\ket{00}+\ket{11}\right)/\sqrt{2}$ 
		\item  $A_0=Z$, $A_1=\left(X+Z\right)/\sqrt{2}$, $A_2=X$, 
		
		$A_3=\left(X-Z\right)/\sqrt{2}$
		\item  $B_0=Z$, $B_1=\left(X+Z\right)/\sqrt{2}$, $B_2=X$, 
		
		$B_3=\left(X-Z\right)/\sqrt{2}$
	\end{itemize}
		
	\item Alice has two measurement settings, whereas Bob has three. The target distribution is generated by:
	\begin{itemize}
		\item $\ket{\phi^+}=\left(\ket{00}+\ket{11}\right)/\sqrt{2}$ 
		\item  $A_0=Z$, $A_1=X$
		\item  $B_0=Z$, $B_1=\left(X+Z\right)/\sqrt{2}$, 
		
		$B_2=\left(X-Z\right)/\sqrt{2}$
	\end{itemize}
	
	\item Both Alice and Bob have two possible measurement settings. The target distribution is generated by:
	\begin{itemize}
		\item $\ket{\phi^+}=\left(\ket{00}+\ket{11}\right)/\sqrt{2}$ 
		\item  $A_0=Z$, $A_1=X$
		\item  $B_0=\left(X+Z\right)/\sqrt{2}$, $B_1=\left(X-Z\right)/\sqrt{2}$
	\end{itemize}
\end{enumerate}

\begin{figure}
\subfloat[]{\label{fig:4_4}
\includegraphics[width=0.45\textwidth]{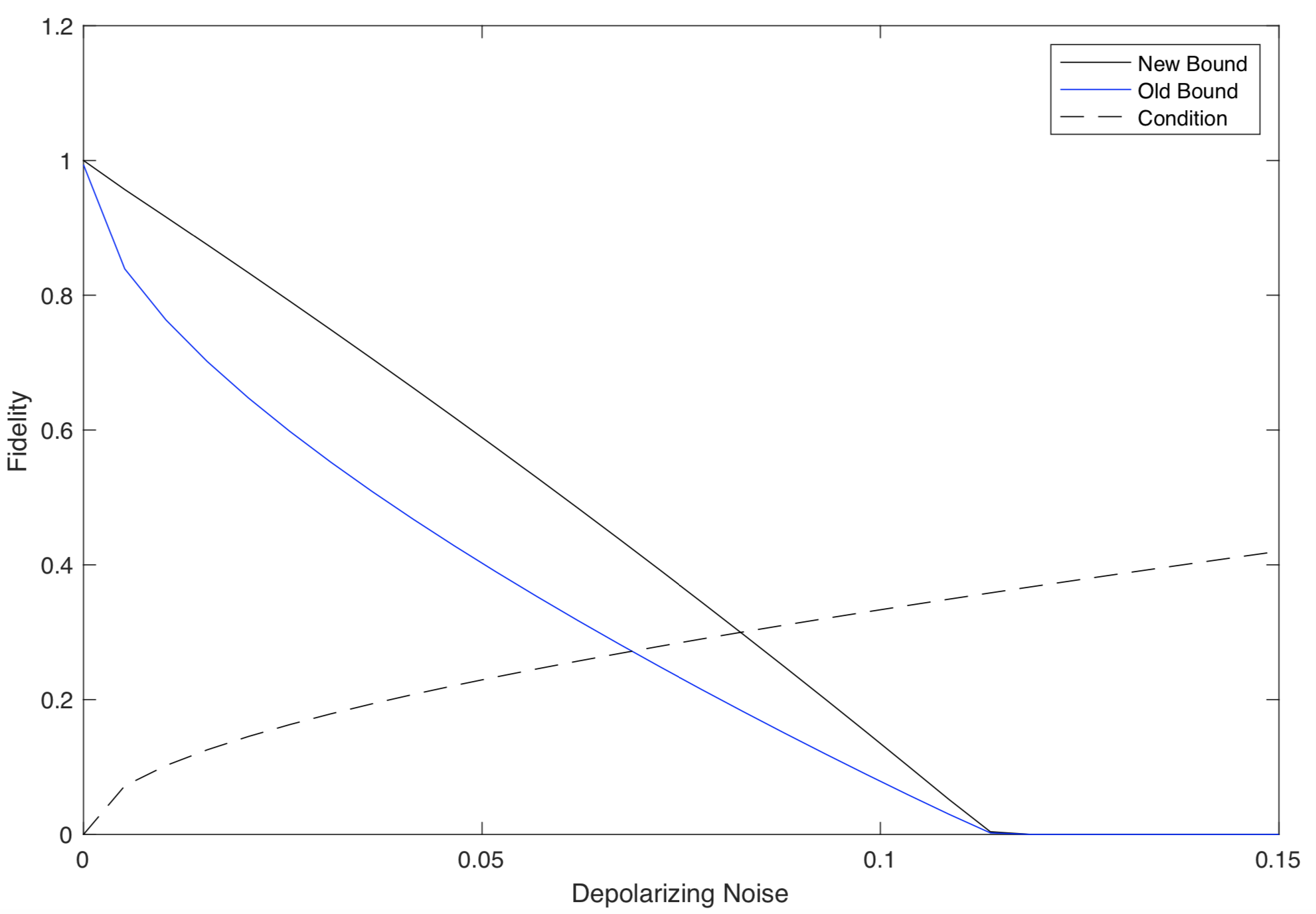}
} \\
\subfloat[]{\label{fig:3_2}
\includegraphics[width=0.45\textwidth]{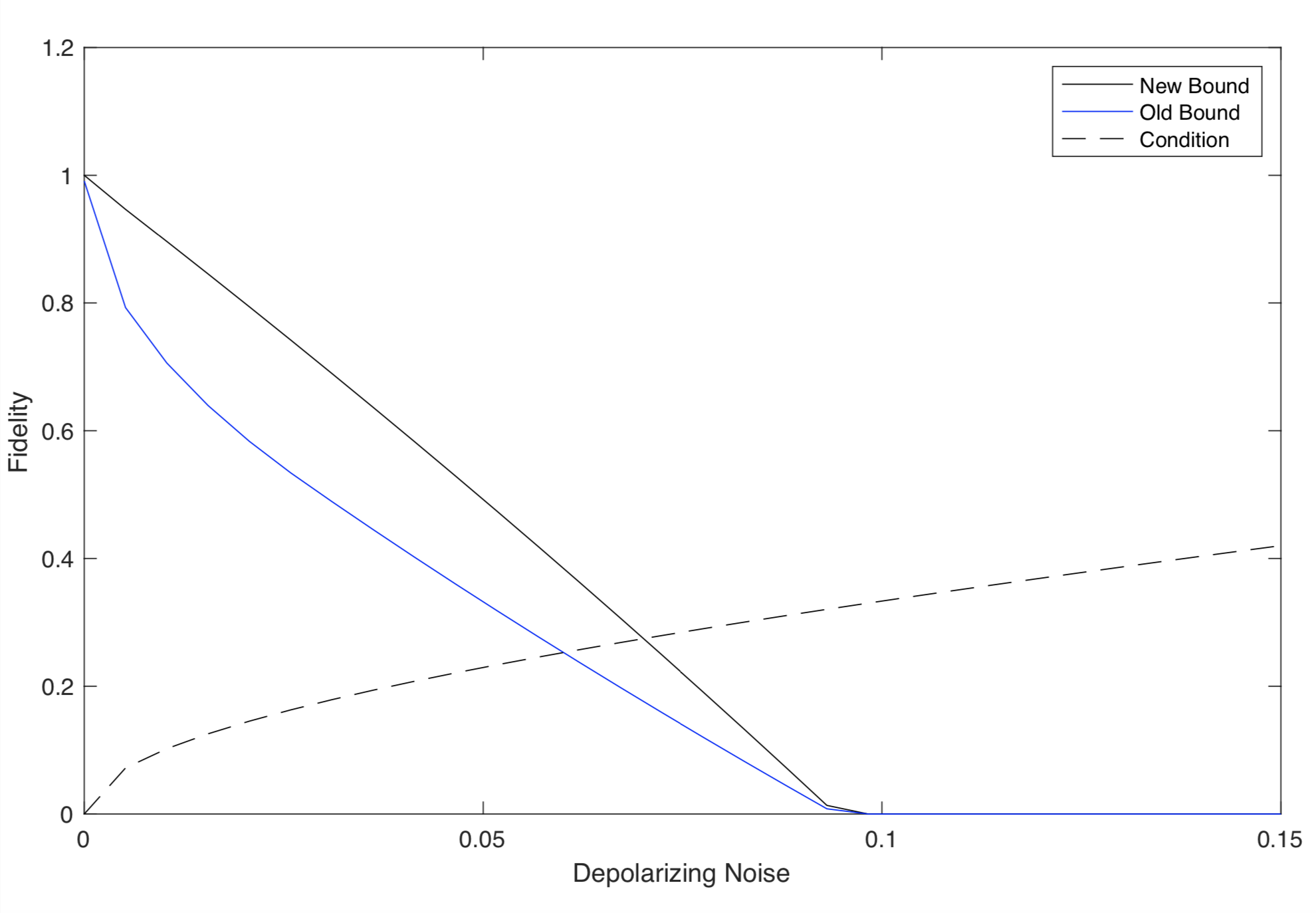}
} \\
\subfloat[]{\label{fig:2_2}
\includegraphics[width=0.45\textwidth]{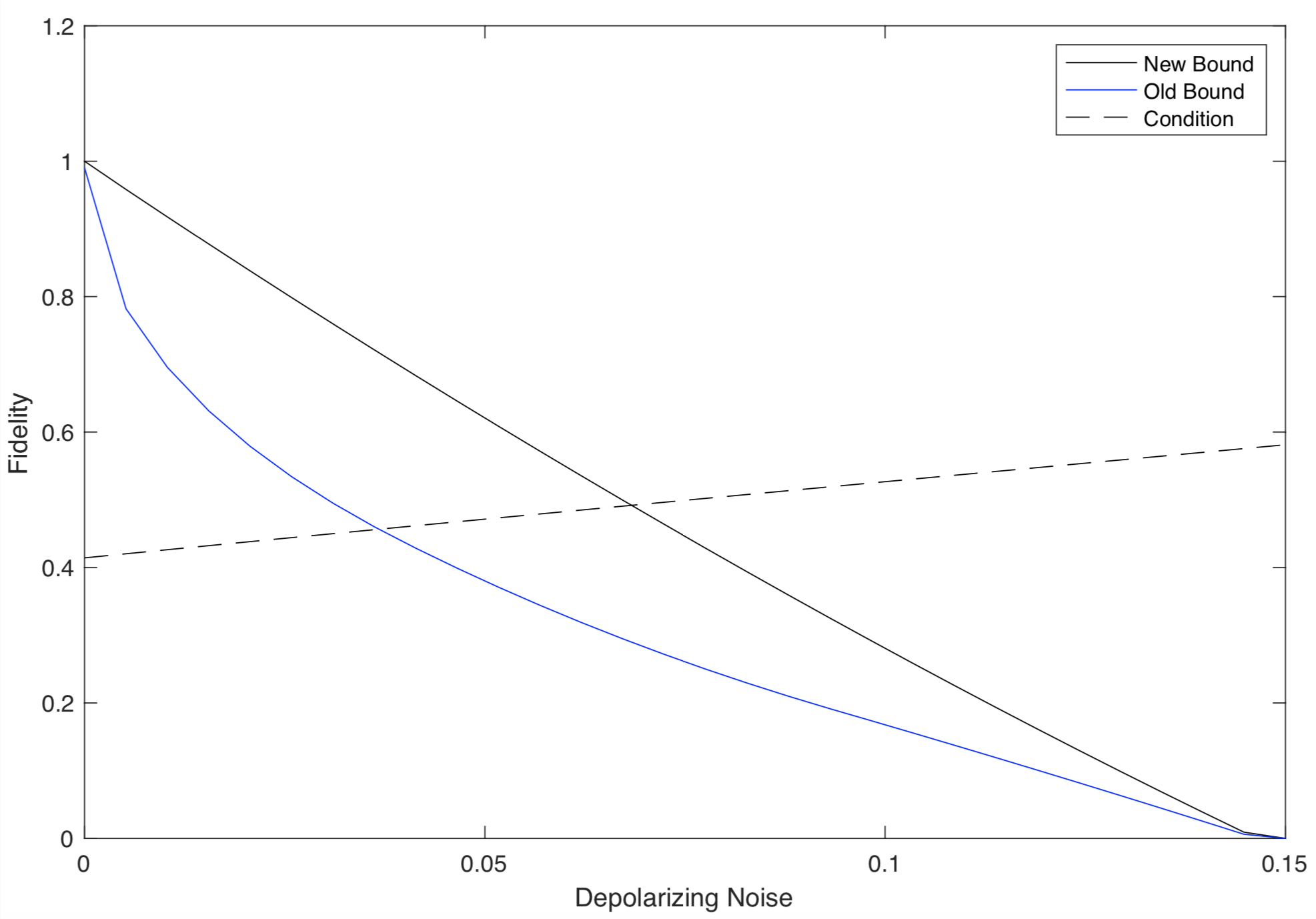}
}
\caption{Fidelity bounds as a function of depolarizing noise in several scenarios (described in main text), regarding the repetition-code protocol defined in Sec.~\ref{sec:protocol}. The blue and black solid lines respectively represent the fidelity bounds derived via the previous approach (based on the Fuchs--van de Graaf inequality) and our new algorithm. It can be seen that the latter yields substantially better bounds. The dashed lines show the value of $\sqrt{\epsilon/(1-\epsilon)}$, so the points where they intersect the solid lines give the threshold values for which advantage distillation is possible according to the condition~\eqref{eq:sufffidcond1}. For our approach, these thresholds  are $q \approx 8.3 \% $, $q \approx 7.0 \% $, and $q \approx 6.7 \% $ (from top to bottom in the scenarios shown here).
\label{fig:grouped}
}
\end{figure}

Case (a) is meant to include the Mayers-Yao self-test \cite{10.5555/2011827.2011830} and measurements that maximize the CHSH value. Alternatively, it can be viewed as having both parties perform all four of the measurements from (c). The results of \cite{tan2019advantage} were able to prove security of this advantage distillation protocol up to $q \approx 6.8 \% $ in this case. We manage to improve the noise tolerance to $q \approx 8.3 \% $, which represents an increase of $1.5 \% $.

As for case (b), the measurements $A_0, A_1, B_1, B_2$ maximize the CHSH value, and the key-generating measurement $B_0$ is chosen such that the QBER is zero in $\operatorname{Pr_{target}}$. Again, our approach allows us to improve the noise tolerance threshold from $q \approx 6.0 \% $ to $q \approx 7.0 \% $.

The final case is a simple CHSH-maximizing setup, where both parties only have two measurements (this is similar to (b), but without the QBER-minimizing measurement for Bob). If we apply Theorem~\ref{Theoremsufffid1} for this case, our approach improves the threshold from $q \approx 3.6 \% $ to $q \approx 6.7 \% $. Also, if we instead optimize the measurements for robustness against depolarizing noise, this threshold can be increased to $q \approx 7.6 \% $ (these optimized measurements correspond to measurements in the $x$-$z$ plane at angles $\theta_{A_0}=0$, $\theta_{A_1} \approx 4.50$, $\theta_{B_0}\approx 3.61$, and $\theta_{B_1}\approx 5.39$ from the $z$-axis).

However, it is important to note that in case (c), Theorem~\ref{Theoremsuffdist} could be applied instead to yield a noise tolerance bound of $q \approx 7.7 \% $, or $q \approx 9.1 \% $ with optimized measurements~\cite{tan2019advantage}. These values are higher than those obtained above by applying our approach to case (c) with Theorem~\ref{Theoremsufffid1}. This gives strong evidence that the sufficient condition in Theorem~\ref{Theoremsufffid1} is not a necessary one, because our approach should yield threshold values close to the optimal ones that could be obtained based only on that sufficient condition. {(In principle, there may still be a gap between the true fidelity values and the bounds we computed; however, we consider this somewhat unlikely --- see Sec.~\ref{sec:methodsSDPalg}.)} Furthermore, we remark that it also suggests that the states $\rho_{E|00},\rho_{E|11}$ that minimize the fidelity in this scenario cannot be pure. This is because in~\cite{tan2019advantage}, the critical inequalites used in the proof of Theorem~\ref{Theoremsufffid1} are all saturated (in the large-$n$ limit, at least) if those states are pure, indicating that the resulting sufficient condition should be basically `tight' (i.e.~a necessary condition) in this case. Since our results indicate that the sufficient condition is not a necessary one, this implies that the relevant states $\rho_{E|00},\rho_{E|11}$ cannot be pure.

Given the above observations, we now conjecture what might be a necessary condition for security of the repetition-code protocol, to serve as a counterpart to Theorem~\ref{Theoremsufffid1}. 

\subsection{Conjectured Necessary Condition for Secret Key Distillation}

While our result in this section can be stated in terms of the fidelity, we note that the reasoning holds for any distinguishability measure $\overlap(\rho,\sigma)$ that has the following properties:
\begin{gather}
\overlap(\rho,\sigma) = \overlap(\sigma,\rho) \ , \label{eq:symm}\\
\overlap(\rho \otimes \rho',\sigma \otimes \sigma') = \overlap(\rho,\sigma) \overlap(\rho',\sigma') \ , \label{eq:mult}\\
\overlap(\rho,\sigma) + d(\rho,\sigma) \geq 1 \ , \label{eq:FvdGlow} \\
\overlap(\rho,\sigma)^2 + d(\rho,\sigma)^2 \leq 1 \ . \label{eq:FvdGupp}
\end{gather}
The first line states that it is symmetric, the second that it is multiplicative across tensor products, and the last two lines correspond to the Fuchs--van de Graaf inequality in the case of the fidelity. An example of another distinguishability measure that satisfies these properties is the \emph{pretty-good fidelity}~\cite{PhysRevA.69.032106, Aud12, 7782776}, defined as $\operatorname{F}_\mathrm{pg}(\rho, \sigma) \coloneqq \operatorname{Tr} \left[ \sqrt{\rho} \sqrt{\sigma} \right]$. To keep our result more general, we shall first present it in terms of any such measure $\overlap$, and discuss at the end of this section which choices yield better bounds. (In fact, we will only need~\eqref{eq:symm}--\eqref{eq:FvdGlow} for this section. We list \eqref{eq:FvdGupp} as well because it can be used when studying \emph{sufficient} conditions; see Sec.~\ref{sec:necproof}.)

For any such $\overlap$, we shall show that under a particular assumption, the condition
\begin{equation}  \label{eq:necfidcond1}
\overlap(\rho_{E|00}, \rho_{E|11}) > \frac{\epsilon}{(1-\epsilon)}
\end{equation}
is necessary for the repetition-code protocol as described above to achieve positive asymptotic keyrate.
(Note that this claim is restricted to the specific protocol description above; in particular, we require that after the initial blockwise `distillation' procedure, only one-way error correction from Alice to Bob is performed, without any further processing such as another iteration of the blockwise `distillation'.) 
Specifically, the assumption is that given some state and measurements compatible with the observed statistics, Eve can produce some other state and measurements (with the same measurement-outcome probability distribution) that have the same value of $\overlap(\rho_{E|00}, \rho_{E|11})$, but with $\rho_{E|01} = \rho_{E|10}$.
This assumption seems reasonable because $\rho_{E|01} = \rho_{E|10}$ describes a situation where Eve is unable to distinguish the cases where Alice and Bob's outcomes are $01$ versus $10$, i.e.~in some sense this appears to be a `suboptimal' attack by Eve. \markup{(It might seem that this could be trivially satisfied by having Eve erase her side-information conditioned on Alice and Bob obtaining different outcomes. However, for Eve to do so, she would need to know when Alice and Bob obtain different outcomes, and it is not clear that this is always possible for the states in Eve's optimal attacks in DIQKD scenarios. This property does hold for the QKD protocols studied in~\cite{bae2006key,PhysRevA.73.012327}, and they use it as part of the proof that their condition is also necessary.)}

To prove that \eqref{eq:necfidcond1}
is in fact necessary (given the stated assumption), we show that if it is not satisfied, then regardless of the choice of block size in the repetition-code protocol, in each block Eve can always produce a classical bit $C''$ such that
\begin{equation}\label{eq:neccondBSC}
\operatorname{H}(C |C'';D=1) \leq \operatorname{H}(C |C';D=1).
\end{equation}
This in turn implies (as noted in~\cite{PhysRevA.73.012327}, using the results for binary symmetric channels in~\cite{256484}) that the repetition-code protocol as described above cannot achieve positive asymptotic keyrate. To prove that Eve can indeed do this, we first derive an `intermediate' implication (we give the proof in Sec.~\ref{sec:necproof}, based on the arguments in \cite{tan2019advantage,bae2006key,PhysRevA.73.012327}): 
\begin{prop} \label{prop:neccondprop1}
Let $\overlap$ be a function satisfying~\eqref{eq:symm}--\eqref{eq:FvdGlow}.
If Eve's side-information satisfies $\rho_{E|01} = \rho_{E|10}$ and
\begin{equation}  \label{eq:nnecfidcond2}
\overlap(\rho_{E|00}, \rho_{E|11}) \leq \frac{\epsilon}{(1-\epsilon)} \ ,
\end{equation}
then for all $n$, Eve can use the information available to her (i.e.~$\mathbf{E}$, $\mathbf{M}$, and $D=1$) to construct a guess $C''$ for the bit $C$ that satisfies
\begin{equation} \label{eq:nprobcond1}
\operatorname{Pr}(C \neq C''|D=1) \leq \operatorname{Pr}(C \neq C'|D=1) \ .
\end{equation} 
\end{prop}

We now observe that the following relations hold (the first by Fano's inequality, the second by a straightforward calculation using the fact that the bits $C,C'$ have uniform marginal distributions in the accepted blocks):
\begin{equation} 
\begin{aligned}
\operatorname{H}(C |C'';D=1) &\leq h_2(\operatorname{Pr}(C \neq C''|D=1)) \\ 
\operatorname{H}(C |C';D=1) &= h_2(\operatorname{Pr}(C \neq C'|D=1)) \ ,
\end{aligned}
\end{equation} 
where $h_2$ is the binary entropy function. With this we see that Eve can indeed produce a bit such that~\eqref{eq:neccondBSC} holds, thereby concluding the proof.

If it could be shown that $\rho_{E|01} = \rho_{E|10}$ is in fact always possible to achieve for Eve without compromising $\overlap(\rho_{E|00}, \rho_{E|11})$ or the post-measurement probability distribution, then \eqref{eq:necfidcond1} would genuinely be a necessary condition for security. It could then be used to find upper bounds for noise tolerance of the repetition-code protocol.

Regarding specific choices of the measure $\overlap$, note that the pretty-good fidelity and the fidelity are related by $\operatorname{F}(\rho, \sigma) \geq \operatorname{F}_\mathrm{pg}(\rho, \sigma) \geq \operatorname{F}(\rho, \sigma)^2$~\cite{Aud12}, with the first inequality being saturated for commuting states and the second inequality for pure states. In particular, the $\operatorname{F}(\rho, \sigma) \geq \operatorname{F}_\mathrm{pg}(\rho, \sigma)$ side of the inequality implies that when aiming to find upper bounds on noise tolerance of this protocol using the condition~\eqref{eq:necfidcond1}, it is always better to consider pretty-good fidelity rather than fidelity (since all states satisfying the condition~\eqref{eq:necfidcond1} with $\overlap=\operatorname{F}_\mathrm{pg}$ also must satisfy it with $\overlap=\operatorname{F}$, but not vice versa). We leave for future work the question of whether there are other more useful choices for the measure $\overlap$.

\section{Discussion}
We discuss several consequences of our findings in this section. With regards to noise tolerances, our results show that by calculating bounds directly via the fidelity, a significant improvement can be achieved over the results based on Theorem~\ref{Theoremsufffid1} in~\cite{tan2019advantage}. This is especially important for the 3-input-scenario for Bob, as the previous advantage distillation bound falls well short of the bound for standard one-way error correction (even when only accounting for the QBER and CHSH value) \cite{pironio2009deviceindependent,woodhead2020deviceindependent,sekatski2020deviceindependent}. As advantage distillation improves the key rate for device-dependent QKD \cite{1176619,PhysRevA.66.060302,renner2005security,bae2006key,PhysRevA.73.012327,Watanabe_2007,Khatri2017NumericalEF}, it is expected to behave analogously for the device-independent setting. Our bound of $q \approx 7.0 \% $ lies within $0.15 \% $ of the noise threshold for one-way protocols using the CHSH inequality \cite{pironio2009deviceindependent} and within $0.40 \% $ of those using asymmetric CHSH inequalities \cite{woodhead2020deviceindependent,sekatski2020deviceindependent}, thereby substantially reducing this gap; however, it still does not yield an overall improvement. Again, this suggests that the sufficient condition~\eqref{eq:sufffidcond1} is in fact not necessary.

This brings us to the question of finding a condition that is both necessary and sufficient for security (of the repetition-code protocol) in the DIQKD setting. Note that while one could view the results of~\cite{bae2006key,PhysRevA.73.012327} as stating that condition~\eqref{eq:sufffidcond1} is both necessary and sufficient in the device-dependent QKD scenarios studied there, there are some subtleties to consider. Namely, that condition could be rewritten in several ways that are equivalent in those QKD scenarios, but not necessarily in DIQKD. For instance, in those QKD scenarios the states $\rho_{E|ab}$ are all pure, which means that $\operatorname{F}(\rho_{E|00}, \rho_{E|11})^2$ could be rewritten as $1-d(\rho_{E|00}, \rho_{E|11})^2$; however, this equivalence may not hold in the DIQKD setting, where those states might not be pure in general. Hence if we think in terms of trying to extend the necessary and sufficient condition in~\cite{bae2006key,PhysRevA.73.012327} to DIQKD, we would first need to address the question of finding the `right' way to formulate that condition.

Indeed, our findings raise the question of whether attempting to determine security via $\operatorname{F}(\rho_{E|00}, \rho_{E|11})$ is the right approach at all, because of the following informal argument. First, assuming that the fidelity bounds we obtained were essentially tight, our results indicate that there are scenarios where condition~\eqref{eq:sufffidcond1} is violated but the repetition-code protocol is still secure (via Theorem~\ref{Theoremsuffdist} instead), implying that it is not a necessary condition. However, if it is not necessary, it is not immediately clear how one might improve upon it.
In particular, it seems unlikely that our conjectured necessary condition~\eqref{eq:necfidcond1} could also be sufficient --- after all, for the device-dependent QKD protocols studied in~\cite{bae2006key,PhysRevA.73.012327}, it is condition~\eqref{eq:sufffidcond1} rather than~\eqref{eq:necfidcond1} that is both sufficient and necessary. 
Since a Fuchs--van de Graaf inequality was used to incorporate the fidelity into the security condition and this inequality is most likely the reason that~\eqref{eq:sufffidcond1} is not necessary, finding a new, completely device-independent approach might necessitate using different inequalities. 

A speculative, but interesting, alternative approach could be to instead consider the (non-logarithmic) quantum Chernoff bound \cite{Audenaert_2007,10.1214/08-AOS593}, 
\begin{equation}
Q(\rho,\sigma) \coloneqq \inf_{0 \leq s \leq 1} \operatorname{Tr} \left[ \rho^s \sigma^{1-s} \right] \ .
\end{equation}
This is because this measure yields asymptotically tight bounds on the distinguishability of the states $\rho^{\otimes n}$ and $\sigma^{\otimes n}$ \cite{10.1214/08-AOS593}, 
which we might be able to use.
However, there is still some work that needs to be done before noise tolerances can be calculated via this method. For example, in contrast to the fidelity, it still not known whether there exists a measurement that preserves this distinguishability measure (this is the main reason we could construct an SDP for bounding the fidelity). Moreover, a security condition for the repetition-code protocol will in all likelihood require a measure of distinguishability between states that are not of the form $\rho^{\otimes n}$ and $\sigma^{\otimes n}$, albeit with some similarities (see Sec.~\ref{sec:necproof}). Hence one would need to investigate which aspects of the proof in~\cite{10.1214/08-AOS593} could be generalized to such states as well. We discuss this further in Sec.~\ref{sec:necproof}.

Similar to \cite{pironio2009deviceindependent}, one could conduct a qubit analysis in the hopes that this approach produces fidelity bounds that are `strong enough'. We show in \cite{HahnThesis}, however, that this is not the case. Moreover, we prove for maximal CHSH violation, i.e. $S=2\sqrt{2}$, that $\operatorname{F}(\rho_{E|00}, \rho_{E|11})=1$ must hold for qubit strategies. As can be seen in \cite{HahnThesis}, this will no longer generally be the case in higher dimensions.

In principle, by combining the bounds we computed here with the security proof in~\cite{tan2019advantage}, one could compute lower bounds on the keyrates under the I.I.D. assumption (both in the asymptotic limit and for finite sample sizes, by using the finite version of the quantum asymptotic equipartition property). However, some numerical estimates we performed indicate that the resulting values are very low, even in the asymptotic case. Informally, this is likely because the proof in~\cite{tan2019advantage} bounds the von Neumann entropy in terms of fidelity with an inequality that is suboptimal in this context, as previously discussed. However, in the case of device-dependent QKD, the keyrates of this protocol are more reasonable, so it is possible that with better proof techniques, the same may hold for DIQKD.

We conclude by mentioning that it should be possible to use our algorithm to obtain keyrate lower bounds for one-way communication protocols as well. This is because the keyrate in such protocols is given by~\cite{Devetak_2005} 
\begin{equation}
\operatorname{H}(A_0|E) - \operatorname{H}(A_0|B_0) \ ,
\end{equation}
and the main challenge in computing this value is finding lower bounds on $\operatorname{H}(A_0|E)$ (again, the $\operatorname{H}(A_0|B_0)$ term is straightforward to handle, e.g.~by estimating the QBER). Taking $A_0$ to be symmetrized as previously mentioned, we can apply the following inequality~\cite{PhysRevLett.105.040505}:
\begin{equation} \label{eq:entfidbnd}
\operatorname{H}(A_0|E) \geq 1-h_2\left(\frac{1-\operatorname{F}(\rho_{E|0}, \rho_{E|1})}{2}\right) \ ,
\end{equation}
where the states $\rho_{E|a_0}$ refer to conditioning on the outcome $A_0$ only. Hence $\operatorname{H}(A_0|E)$ can be bounded in terms of $\operatorname{F}(\rho_{E|0}, \rho_{E|1})$. It should be straightforward to adapt our algorithm to bound such a fidelity expression as well, in which case our approach would also be useful to compute keyrates for non-advantage distillation setups. We aim to investigate this in future work.

\section{Methods}
\subsection{Fidelity For Post-Measurement Tripartite States}
\begin{prop}
The fidelity, $\operatorname{F}(\rho_{E|00},\rho_{E|11})$, of a state of the form \eqref{eq:Initialstate} is given by 
\begin{equation}
\operatorname{F}(\rho_{E|00},\rho_{E|11}) =\sum_{i} \frac{\sqrt{ \operatorname{Pr}(00| i) \operatorname{Pr}(11| i) }}{\operatorname{Pr}(00)}  \operatorname{Pr}(i) \ .
\end{equation}
\end{prop}
\begin{proof}
It is easily verified that:
\begin{align}
\rho_{E|00} &= \sum_{i} \frac{\operatorname{Pr}(00| i) \operatorname{Pr}(i) }{\operatorname{Pr}(00)} \ketbra{i}{i} \\
\rho_{E|11} &= \sum_{i} \frac{\operatorname{Pr}(11| i) \operatorname{Pr}(i) }{\operatorname{Pr}(11)} \ketbra{i}{i} \ .
\end{align}
We note that $\operatorname{Pr}(00)=\operatorname{Pr}(11)$ due to the symmetrization step. Since these states are diagonal in the same basis, 
we can directly compute the fidelity:
\begin{align}
&\operatorname{F}(\rho_{E|00},\rho_{E|11}) \\
=& \operatorname{Tr}\left(\sqrt{ \rho_{E|11}^{\frac{1}{2}} \rho_{E|00}  \rho_{E|11}^{\frac{1}{2}}}\right) \\
=& \operatorname{Tr}\left( \sqrt{\sum_{i} \frac{\operatorname{Pr}\left(11| i \right) \operatorname{Pr}\left(00| i\right) \operatorname{Pr}\left(i\right)^2 }{\operatorname{Pr}\left(00\right)^2} \ketbra{i}{i}}\right) \\
=& \sum_{i} \frac{\sqrt{ \operatorname{Pr}(00| i) \operatorname{Pr}(11| i) }}{\operatorname{Pr}(00)}  \operatorname{Pr}(i) \ ,
\end{align}
as claimed.
\end{proof}
\subsection{Approximating The Fidelity With Polytope Hyperplanes} \label{Sec:Polytopes}

Our goal in this section will be to find a lower bound on $\sqrt{ \operatorname{Pr}(00| i) \operatorname{Pr}(11| i) }$ that can be written as a pointwise minimum of affine functions (of $\vect{p}^i$). Note that this function depends only
on $\operatorname{Pr}(00|i)$ and $\operatorname{Pr}(11|i)$. (We need to consider both $\operatorname{Pr}(00|i)$ and $\operatorname{Pr}(11|i)$ without assuming they are equal, as we cannot assume the distributions conditioned on $i$ are still symmetrized. This is because Eve's measurement takes place after the symmetrization step, and does not have to respect the symmetry.) Hence for the purposes of this section, we shall focus on 2-dimensional vectors $\vect{x}$, with the implicit understanding that 
\begin{equation}
\vect{x} = \begin{pmatrix}
x_1 \\
x_2 
\end{pmatrix} = \begin{pmatrix}
\operatorname{Pr}(00|i) \\
\operatorname{Pr}(11|i)
\end{pmatrix} \ .
\end{equation}
An affine function of such a vector $\vect{x}$ straightforwardly defines a corresponding affine function of $\vect{p}^i$, by considering the latter to depend only on the terms $\operatorname{Pr}(00|i)$ and $\operatorname{Pr}(11|i)$.

We begin the construction by defining an appropriate lattice of points:
\begin{defn}
For each $n \in \mathbb{N}$, let $\mathcal{L}_{n}$ denote a uniformly spaced grid of $\left(2^n+1\right) \cdot \left(2^n +1\right)$ points in $[0,1]^2$, i.e.
\begin{equation}
\mathcal{L}_{n} \coloneqq \left\{0,\frac{1}{2^n},\frac{2}{2^n},\dots,1\right\}^2 \ .
\end{equation}
For a concave function $f:[0,1]^2 \to \mathbb{R}$, we define the \emph{$\textit{n}^{th}$ order lattice of $f$} to be the set of $\left(2^n+1\right) \cdot \left(2^n +1\right)$ ordered triples $(x_1,x_2,f(\vect{x}))$ such that $(x_1,x_2) \in \mathcal{L}_{n}$.
\end{defn}
\noindent (The above definition naturally generalizes to higher-dimensional arrays.) 
\begin{prop} \label{propconvhull}
Let $\ptope_n$ denote the convex hull of the union of the $n^{th}$ order lattice of $f(\vect{x})\coloneqq\sqrt{x_1 x_2}$ and the $n^{th}$ order lattice of $g(\vect{x})\coloneqq0$. Then $\ptope_n$ is a convex polytope that lies on or beneath the graph of the function $f(\vect{x})=\sqrt{x_1 x_2}$. 
\end{prop}
\begin{proof}  
$ $ \newline
\noindent \underline{Step 1}: Show that the $2 \cdot (2^n+1)^2$ lattice points lie on or beneath the graph of $f(\vect{x})$.

\noindent The lattice points of $f(\vect{x})$ automatically lie on the function graph as desired. Also, since $f(\vect{x})$ is always non-negative, the lattice points of $g(\vect{x})$ also lie on or beneath the function graph. 

\noindent \underline{Step 2}: Show that the convex hull of the union of both lattices lies on or beneath the graph of $f(\vect{x})$. 

\noindent All points in the convex hull are a convex sum of the lattice points. As $f(\vect{x})$ is a concave function, such convex sums must lie on or beneath the graph of $f(\vect{x})$ as well. 
\end{proof}

We can use $\ptope_n$ to construct the desired lower bound on $f(\vect{x})$, as follows. {As we shall formally prove in Proposition~\ref{equaldefns}, this process is basically constructing the upper envelope of $\ptope_n$ (i.e.~the function whose graph is the `upper surface' of $\ptope_n$; this is formally defined by~\eqref{eq:fndefmax} below).} 

We first transform $\ptope_n$ to its facet representation~\cite{henk1995basic}, i.e.~the description of the polytope by its facet-defining half-spaces. This gives us a set of inequalities described by parameters $\{(a_j, b_j, c_j, d_j)\}_j$, such that $(x_1, x_2, x_3) \in \ptope_n$ if and only if
\begin{equation} \label{eq:facet}
a_j x_1 + b_j x_2 + c_j x_3 \leq d_j \ \ \forall j \ .
\end{equation}
To only retrieve the facets that will be used to approximate $f(\vect{x})$, i.e.~the facets that describe the upper envelope of the polytope we constructed, keep only the facets for which $c_j > 0$. Geometrically, this corresponds to facets such that the normal vector (directed outwards from the polytope) has a vertical component that points upwards. 

In our case, this means we remove the facets
\begin{align}
-x_3 &\leq 0 \label{eq:removed1} \\
x_1 &\leq 1 \label{eq:removed2} \\ 
x_2 &\leq 1 \ , \label{eq:removed3}
\end{align}
which correspond to a lower horizontal facet and two vertical facets, respectively. (For more general concave $f$, there would be more vertical facets to remove, but the subsequent analysis still holds as it is based only on the fact that we keep exactly the facets with $c_j>0$.)
Let $\mathcal{S}$ denote the set of indices of the remaining facets. For each $j\in\mathcal{S}$, we define a corresponding affine function,
\begin{equation}
h_j(\vect{x}) \coloneqq \frac{1}{c_j}(d_j - a_j x_1 - b_j x_2 ) \ ,
\end{equation}
and use these to define a function (denoted $f_{n}$) that is meant to bound $f(\vect{x})$:
\begin{equation} \label{eq:fndefmin}
f_{n}(\vect{x}) \coloneqq \min_{j\in\mathcal{S}} h_j(\vect{x}) \ .
\end{equation}

We verify that the above procedure indeed produces the upper envelope of $\ptope_n$:

\begin{prop}\label{equaldefns}
For all $n \in \mathbb{N}$ and $\vect{x} \in \left[ 0,1 \right]^2 $, $f_{n}$ as defined above satisfies
\begin{align} 
f_{n}(\vect{x}) = \max \{x_3 \mid (\vect{x},x_3)\in\ptope_n \} \ . \label{eq:fndefmax} 
\end{align}
\end{prop}
\begin{proof}
We first remark that it is indeed valid to write the expression~\eqref{eq:fndefmax} as a maximum rather than a supremum, because by construction of $\ptope_n$, the feasible set in~\eqref{eq:fndefmax} is non-empty (for $\vect{x} \in \left[ 0,1 \right]^2$) and compact. 

To prove the desired equality, we start by considering a fixed $\vect{x} \in \left[ 0,1 \right]^2$, and arguing that the expression~\eqref{eq:fndefmax} is in fact equal to 
\begin{align} 
\max \{x_3 \mid a_j x_1 + b_j x_2 + c_j x_3 \leq d_j \ \ \forall j \in \mathcal{S} \} \ , \label{eq:fndefrelax} 
\end{align}
where the values $a_j,b_j,c_j,d_j$ are from the facet inequalities~\eqref{eq:facet}. This is because $\ptope_n$ is exactly the set of points which satisfy the facet inequalities~\eqref{eq:facet} for all $j$, so the only difference between~\eqref{eq:fndefmax} and~\eqref{eq:fndefrelax} is that the latter maximization has omitted the facet inequalities such that $c_j \leq 0$. Removing these inequalities does not change the maximum value, by the following argument. The inequalities with $c_j=0$ are independent of $x_3$, so either they are satisfied for all $x_3$ or for no $x_3$; however, as previously noted the feasible set of~\eqref{eq:fndefmax} is non-empty (for $\vect{x} \in \left[ 0,1 \right]^2$), so the former must be the case. This implies that removing them does not change the maximum value. As for the inequalities with $c_j < 0$, notice that they are \emph{lower} bounds on $x_3$, which means that removing them also does not change the maximum value (as long as the original maximization~\eqref{eq:fndefmax} is feasible, which it indeed is as noted previously). Thus the expressions~\eqref{eq:fndefmax} and~\eqref{eq:fndefrelax} are equal for any $\vect{x} \in \left[ 0,1 \right]^2$.

It remains to show that the original definition~\eqref{eq:fndefmin} of $f_{n}$ is equal to~\eqref{eq:fndefrelax} (when treating the latter as a function of $\vect{x}$ on the same domain). To do so, we show that they have the same subgraph. The subgraph of~\eqref{eq:fndefmin} is 
\begin{align}
& \{(\vect{x},x_3) \mid x_3 \leq h_j(\vect{x}) \ \ \forall j \in \mathcal{S} \} \\
=& \left\{(\vect{x},x_3) \,\middle|\, x_3 \leq \frac{1}{c_j}(d_j - a_j x_1 - b_j x_2 ) \ \ \forall j \in \mathcal{S} \right\} \\
=& \{(\vect{x},x_3) \mid a_j x_1 + b_j x_2 + c_j x_3 \leq d_j \ \ \forall j \in \mathcal{S} \} \ ,
\end{align}
using the fact that $c_j>0$ for all $j\in\mathcal{S}$. (To be precise, in those expressions $\vect{x}$ should be restricted to the function domain, but the argument at this step holds regardless of whether we take the domain to be $\left[ 0,1 \right]^2$ or $\mathbb{R}^2$.) The last line is the subgraph of~\eqref{eq:fndefrelax}, so indeed the functions are equal. 
\end{proof}

With the formula~\eqref{eq:fndefmax}, we can prove some intuitive properties of $f_{n}$, which will be useful in subsequent arguments:

\begin{prop}  \label{lbound}
\begin{equation} \label{eq:lbound}
f_{n}(\vect{x}) \leq f(\vect{x}) 
\ \ \forall n \in \mathbb{N} \ \ \forall \vect{x} \in \left[ 0,1 \right]^2 
\ ,
\end{equation}
with $f_{n}(\vect{x}) = f(\vect{x})$ whenever $\vect{x} \in \mathcal{L}_{n}$.
\end{prop}
\begin{proof}
We showed in Proposition~\ref{propconvhull} that $\ptope_n$ lies on or below the graph of $f$. Hence~\eqref{eq:lbound} follows immediately from the formula~\eqref{eq:fndefmax}. 

As for the equality condition, we note that for all $\vect{x} \in \mathcal{L}_{n}$ we have that $(\vect{x}, f(\vect{x}))$ lies in $\ptope_n$ (by construction). Hence~\eqref{eq:fndefmax} implies $f_{n}(\vect{x}) \geq f(\vect{x})$, which implies they must be equal since the reverse inequality~\eqref{eq:lbound} holds in general.
\end{proof}

\begin{prop} \label{monconvergence}
\begin{equation}
f_{n}(\vect{x}) \leq f_{n+1}(\vect{x}) 
\ \ \forall n \in \mathbb{N} \ \ \forall \vect{x} \in \left[ 0,1 \right]^2 
\ .
\end{equation}
\end{prop}
\begin{proof}  
By increasing $n$ by one, we don't remove any lattice points, but rather just add additional points by halving the intervals in both directions. Thus we have ${\ptope_n \subseteq \ptope_{n+1}}$, which yields the desired inequality via~\eqref{eq:fndefmax}.
\end{proof}

We now show the sequence $f_{n}$ indeed converges uniformly to $f$, so it yields arbitrarily tight bounds. The main intuition is that $f_{n}$ forms a monotone sequence of 
uniformly 
continuous functions that have the same value as $f$ on an increasingly fine grid. 

\begin{widetext}
\begin{prop} 
As $n \to \infty$, $f_{n}$ converges uniformly to $f$.
\end{prop}
\begin{proof}
We note that three conditions must hold:
\begin{align}
\forall \delta &> 0 \ \ \exists n \in \mathbb{N}: \ \ \forall \vect{x} \in \left[ 0,1 \right]^2 \ \ \exists \vect{y} \in \mathcal{L}_{n} \ \ \text{s.t.}  \  \ |\vect{x} -\vect{y} | < \delta \label{prop101} \\
 \forall \epsilon_1 &> 0 \ \ \exists \delta_1 >0: \ \ \forall \vect{x},  \vect{y},\in \left[ 0,1 \right]^2\ \ \text{if} \ \  |\vect{x} -\vect{y} | < \delta_1\ \ \text{then} \ \ |f(\vect{x}) -f(\vect{y} )| < \epsilon_1 \label{prop102}\\
\forall \epsilon_2 &> 0 \ \ \exists \delta_2 >0: \ \ \forall \vect{x},  \vect{y},\in \left[ 0,1 \right]^2\ \ \text{if} \ \  |\vect{x} -\vect{y} | < \delta_2\ \ \text{then} \ \ |f_{n}(\vect{x}) -f_{n}(\vect{y} )| < \epsilon_2 \label{prop103}
\end{align}
The first statement just says that one can create an arbitrarily fine grid $\mathcal{L}_{n}$.
The other two statements follow from the fact that continuous functions on compact sets are automatically uniformly continuous.  

To show uniform convergence, we would need to prove that
\begin{gather}
\forall \epsilon > 0 \ \  \exists n \in \mathbb{N}: \ \forall n' \geq n \ \  \forall \vect{x} \in \left[ 0,1 \right]^2 \ , \qquad 
|f_{n^\prime}(\vect{x})-f(\vect{x})| < \epsilon \ .
\end{gather}
To prove this, consider any $\epsilon > 0$, and choose $n$ by the following procedure: set $\epsilon_1 < \epsilon/2$ and $\epsilon_2 < \epsilon/2$, and take some corresponding $\delta_1,\delta_2$ according to~\eqref{prop102}--\eqref{prop103}. Then set $\delta < \min(\delta_1,\delta_2)$, and take a corresponding $n \in \mathbb{N}$ according to~\eqref{prop101}. This choice of $n$ has the desired property: for all $n' \geq n$ and for all $\vect{x} \in \left[ 0,1 \right]^2$,~\eqref{prop101} ensures that there exists some $\vect{y} \in \mathcal{L}_{n}$ satisfying $|\vect{x} -\vect{y} | < \delta$, hence
\begin{align}
|f_{n^\prime}(\vect{x})-f(\vect{x})| &\leq |f_{n}(\vect{x})-f(\vect{x})| \\
&\leq |f_{n}(\vect{x})-f(\vect{y})| + |f(\vect{y})-f(\vect{x})| \\
&= |f_{n}(\vect{x})-f_{n}(\vect{y})| + |f(\vect{y})-f(\vect{x})| \\
&< \epsilon_2 + \epsilon_1 < \epsilon \ ,
\end{align}
where the first line follows from Propositions~\ref{lbound} and~\ref{monconvergence}, while the third line follows from Proposition~\ref{lbound}.
\end{proof}
\end{widetext}

\subsection{Creating An SDP Algorithm That Minimizes The Fidelity} \label{sec:methodsSDPalg}
Without loss of generality, we ignore the factor $1/\operatorname{Pr}(00)$ in the fidelity expression, as it is a positive constant (for a given distribution $\operatorname{Pr}(ab|xy)$). 
We consider the functions $f$ and $f_{n}$ defined in the previous section, but as previously discussed, we now view them as functions of $\vect{p}^i$ (though with dependence only on the $\operatorname{Pr}(00|i),\operatorname{Pr}(11|i)$ terms), i.e.~so we have $f(\vect{p}^i)=\sqrt{ \operatorname{Pr}(00|i) \operatorname{Pr}(11|i) }$ and analogously for $f_{n}$. Since $f_{n}$ is a lower bound on $f$, our optimization problem (after dropping the $1/\operatorname{Pr}(00)$ factor) is clearly lower bounded by the following:
\begin{equation} \label{eq:approxfidopt1}
\begin{aligned}
\inf_{\operatorname{Pr}(i),\vect{p^{i}}} \quad & \sum_{i} f_{n}(\vect{p}^i) \operatorname{Pr}(i) \\
\textrm{s.t.} \quad &\sum_{i} \operatorname{Pr}(i) \vect{p}^{{i}} =\vect{p} \\
\quad &\vect{p}^{i} \in  \mathcal{Q}_{\mathcal{X},\mathcal{Y}} \\
\quad &\operatorname{Pr}(i) \in   \mathcal{P(\mathcal{I})} 
\end{aligned}
\end{equation}

We can show that this lower bound converges uniformly to the original problem~\eqref{eq:exactfidopt2} as $n \to \infty$, using our previous results about convergence of $f_{n}$:
\begin{prop} 
As $n \to \infty$, the optimal value of \eqref{eq:approxfidopt1} converges uniformly to that of \eqref{eq:exactfidopt2} (rescaled by the constant factor of $1/\operatorname{Pr}(00)$).
\end{prop}
\begin{proof}
We denote the solutions to the original and approximate optimization problems by $f^\star(\vect{p})$ and $f_{n}^{\star}(\vect{p})$, respectively. For all $\vect{p} \in  \mathcal{Q}_{\mathcal{X},\mathcal{Y}}$, for all $n \in \mathbb{N}$, and for all $\epsilon_1>0$, there exists a probability distribution $\operatorname{Pr}(i)$ and a set of quantum realizable probability distributions $\vect{p}^i$ that satisfy the optimization constraints, such that
\begin{equation}
\left| \sum_{i} f_{n}(\vect{p}^i) \operatorname{Pr}(i) - f_{n}^{\star}(\vect{p}) \right| < \epsilon_1 \ .
\end{equation}
As $f_{n}(\vect{p}^i)$ converges uniformly to $f(\vect{p}^i)$, for all $\epsilon_2>0$, there exists an $n$ such that for all $n' \geq n$ and all $\vect{p}^i \in  \mathcal{Q}_{\mathcal{X},\mathcal{Y}}$, 

\begin{equation}
\left| \sum_{i} f(\vect{p}^i) \operatorname{Pr}(i)-\sum_{i} f_{n'}(\vect{p}^i) \operatorname{Pr}(i) \right| < \epsilon_2
\end{equation}
If one chooses $\epsilon_1 +\epsilon_2< \epsilon$, then
\begin{gather}
 \forall \epsilon > 0 \ \  \exists n \in \mathbb{N}: \ \forall n' \geq n \ \  \forall \vect{p} \in \mathcal{Q}_{\mathcal{X},\mathcal{Y}} \ , \nonumber\\
\begin{align}
&\left| f^\star(\vect{p})-f_{n'}^{\star}(\vect{p}) \right| \\
\leq& \left|  \sum_{i} f(\vect{p}^i) \operatorname{Pr}(i)-f_{n'}^{\star}(\vect{p}) \right| \\
\leq& \left| \sum_{i} f(\vect{p}^i) \operatorname{Pr}(i)-\sum_{i} f_{n'}(\vect{p}^i) \operatorname{Pr}(i) \right| \nonumber\\
&\quad + \left| \sum_{i} f_{n'}(\vect{p}^i) \operatorname{Pr}(i)-f_{n'}^{\star}(\vect{p}) \right| \\
<& \epsilon_2 + \epsilon_1 < \epsilon \ ,
\end{align}
\end{gather}
as desired. 
\end{proof}

We now describe how to bound~\eqref{eq:approxfidopt1} using SDPs.
\begin{prop} 
\eqref{eq:approxfidopt1} can be solved via an SDP hierarchy.
\end{prop}
\begin{proof}
To reduce the sum to a bounded number of terms, we adopt the approach from \cite{bancal2013randomness,Nieto_Silleras_2014}. Specifically, consider any feasible point of the optimization~\eqref{eq:approxfidopt1}, i.e.~some feasible values for $\vect{p}^{i}$ and $\operatorname{Pr}(i)$. 
Following the previous section, let $\mathcal{S}$ denote the indices of the affine functions used to define $f_{n}$.

Partition the summation domain of $i$ into subsets $\{\mathcal{R}_j\}_{j\in\mathcal{S}}$, such that $i\in \mathcal{R}_j$ implies the minimum in the definition~\eqref{eq:fndefmin} of $f_{n}(\vect{p}^i)$ is attained by the index $j$. In other words, $i\in \mathcal{R}_j$ implies $f_{n}(\vect{p}^i) = h_j(\vect{p}^i)$. (Geometrically speaking, we are partitioning the terms based on which facet of $\ptope_n$ they lie on.) Let us define
\begin{equation}
\tilde{\operatorname{P}}_j = \sum_{i\in\mathcal{R}_j} \operatorname{Pr}(i) \ , \quad 
\tilde{\vect{p}}^j = \sum_{i\in\mathcal{R}_j} \vect{p}^i \frac{\operatorname{Pr}(i) }{\tilde{\operatorname{P}}_j} \label{eq:newQfeas} \ . 
\end{equation}
Note that the choice of partition may not be unique, e.g.~if the feasible point being considered has a $\vect{p}^{i}$ term where the minimum in the definition~\eqref{eq:fndefmin} is attained by more than one $j\in\mathcal{S}$. However, this nonuniqueness is not a problem; any partition with the specified property suffices. Also, some $\mathcal{R}_j$ may be empty, but this is not a problem either; one should simply select an arbitrary distribution $\tilde{\vect{p}}^j \in \mathcal{Q}_{\mathcal{X},\mathcal{Y}}$ instead of using~\eqref{eq:newQfeas} (since the denominator is zero if $\mathcal{R}_j$ is empty).

Then we can rewrite
\begin{align}
\sum_{i\in\mathcal{R}_j} f_{n}(\vect{p}^i) \operatorname{Pr}(i)
&= \sum_{i\in\mathcal{R}_j} h_j(\vect{p}^i) \operatorname{Pr}(i) \\
&= \left(\sum_{i\in\mathcal{R}_j} h_j(\vect{p}^i) \frac{\operatorname{Pr}(i)}{\tilde{\operatorname{P}}_j}\right) \tilde{\operatorname{P}}_j \\
&=h_j(\tilde{\vect{p}}^j) \tilde{\operatorname{P}}_j \label{eq:newobjval} \ ,
\end{align}
where in the third line we used the fact that $h_j$ is affine and $\{{\operatorname{Pr}(i)}/{\tilde{\operatorname{P}}_j}\}_i$ forms a normalized probability distribution over $i \in \mathcal{R}_j$.

Observe that $\tilde{\vect{p}}^j \in \mathcal{Q}_{\mathcal{X},\mathcal{Y}}$ (by convexity of $\mathcal{Q}_{\mathcal{X},\mathcal{Y}}$), and that $\tilde{\operatorname{P}}_j$ is a valid probability distribution (over $j\in\mathcal{S}$) since the sets $\mathcal{R}_j$ partition the sum over $i$. Together with the expression~\eqref{eq:newobjval}, this implies that if we replace the original objective function with 
\begin{equation}  \label{eq:newobjfct}
\sum_{i\in\mathcal{S}} h_i(\vect{p}^i) \operatorname{Pr}(i) \ ,
\end{equation}
the value of the optimization will not increase, since every feasible point of the original optimization yields another feasible point with the same objective value but in the form~\eqref{eq:newobjfct}. In other words, by rewriting the objective function in the form~\eqref{eq:newobjfct}, we have essentially taken $\vect{p}^i$ and $ \operatorname{Pr}(i)$ to be the $\tilde{\vect{p}}^j$ and $\tilde{\operatorname{P}}_j$ we constructed above. Furthermore, since we constructed $f_{n}$ via~\eqref{eq:fndefmin},
\begin{equation}
h_i(\vect{p}^i) \geq f_{n}(\vect{p}^i) 
\end{equation}
holds for all $\vect{p}^{i} \in  \mathcal{Q}_{\mathcal{X},\mathcal{Y}}$. Thus \eqref{eq:newobjfct} is a natural upper bound on our previous objective function in \eqref{eq:approxfidopt1}, so replacing the latter with the former will not decrease the optimal value either. 
In summary, we can replace the objective function with \eqref{eq:newobjfct} without changing the optimal value, which is useful because it is the sum of a (known) finite number of terms; also, we no longer need to address the minimization in the definition of $f_{n}$.

However, since $h_i(\vect{p}^{i}) \operatorname{Pr}(i)$ is a product of affine functions of the optimization variables, \eqref{eq:newobjfct} is still not an affine function. To deal with this, we consider subnormalized probability distributions, i.e.~the terms in the distribution sum up to a value in $[0,1]$ instead of having to sum to $1$ (here we mean summing over $(a,b)$ for each choice of $(x,y)$; also, we impose that all $(x,y)$ have the same normalization factor). In our case, we scale the probability distributions $\vect{p}^i$ by the scaling factor $\operatorname{Pr}(i)$, i.e.~we define new variables $\subn{p}^i = \operatorname{Pr}(i)\vect{p}^i$. 
To verify that the objective function is affine in these new variables, we show that each term in the summation is affine. Writing $h_i(\vect{p}^i)$ in the form $a_i + \vect{a}_i \cdot \vect{p}^i$ for some scalar $a_i$ and vector $\vect{a}_i$, we have
\begin{align}
&h_i(\vect{p}^i) \operatorname{Pr}(i) \\
=& a_i \operatorname{Pr}(i) + (\vect{a}_i \cdot \vect{p}^i) \operatorname{Pr}(i) \\
=& a_i \left(\sum_{ab} \operatorname{Pr}(ab|00i)\operatorname{Pr}(i)\right) + \vect{a}_i \cdot \subn{p}^i \label{eq:sumexpanded} \\
=& a_i \vect{z} \cdot \subn{p}^i + \vect{a}_i \cdot \subn{p}^i \label{eq:sumcollapsed} \ ,
\end{align}
where $\vect{z}$ is a vector, which contains only zeroes and ones, that specifies the terms summed over in~\eqref{eq:sumexpanded} (this is possible since each $\operatorname{Pr}(ab|00i)\operatorname{Pr}(i)$ term is equal to an element of $\subn{p}^i$; also, note that the choice to use the input pair $xy=00$ at that step is arbitrary and any other pair would suffice).
This is indeed affine (in fact linear) in $\subn{p}^i$.
As for the constraints, observe that the first constraint is linear in $\subn{p}^i$. Also, as long as $\vect{p}$ is normalized, we can replace the second and third constraints by a single constraint $\subn{p}^i \in \breve{\mathcal{Q}}_{\mathcal{X},\mathcal{Y}}$, where $\breve{\mathcal{Q}}_{\mathcal{X},\mathcal{Y}}$ denotes subnormalized distributions compatible with quantum theory (and with a common normalization factor for all input pairs $(x,y)$). This is because when $\vect{p}$ is normalized, the first constraint implicitly imposes a normalization condition on the variables $\subn{p}^i$ that subsumes the original third constraint.

In summary, the optimal value of \eqref{eq:approxfidopt1} is the same as
\begin{equation} \label{eq:approxfidoptSDP}
\begin{aligned}
\inf_{\subn{p}^{i}} \quad & \sum_{i\in\mathcal{S}} 
a_i \vect{z} \cdot \subn{p}^i + \vect{a}_i \cdot \subn{p}^i
\\
\textrm{s.t.} \quad &\sum_{i} \subn{p}^i =\vect{p} \\
\quad &\subn{p}^i \in  \breve{\mathcal{Q}}_{\mathcal{X},\mathcal{Y}} 
\end{aligned}
\end{equation}
where $a_i, \vect{a}_i, \vect{z}$ are the values described above regarding~\eqref{eq:sumcollapsed}.
Finally, we use the fact that there exists a SDP hierarchy for the verification of subnormalized quantum probability distributions $\breve{\mathcal{Q}}_{\mathcal{X},\mathcal{Y}}$  \cite{navascues2008convergent,bancal2013randomness,Nieto_Silleras_2014}. Hence, the entire constrained optimization can be lower-bounded by using this hierarchy of SDP relaxations to impose the $\breve{\mathcal{Q}}_{\mathcal{X},\mathcal{Y}}$ constraint, yielding a sequence of increasingly tight lower bounds on the optimization.
Note that each level of the hierarchy yields a certified \emph{lower} bound, i.e.~our results are never an over-estimate of the true minimum of the optimization.
\end{proof}

We close this section with some implementation remarks. For Fig.~\ref{fig:4_4}, we used a $4\times4$ lattice to construct the bound $f_{n}$, and NPA level $2$. For Figs.~\ref{fig:3_2}--\ref{fig:2_2}, we used an $8\times8$ lattice and NPA level $3$ (for the latter case, we found that NPA levels $2$ and $4$ also gave basically the same results). The SDP runtime was not too long in all cases, ranging from a few seconds to under 15 minutes (for each data point on the graphs), depending on the size of the scenarios.

In principle, there are two ways in which our bounds might not be tight: first, we have replaced $f$ with $f_{n}$; second, the SDP hierarchy of~\cite{navascues2008convergent} may not have converged to a sufficiently tight bound. We consider the latter to be less of an issue, because this hierarchy typically performs well in situations with few inputs and outputs (for instance in Fig.~\ref{fig:2_2}, which was the main example supporting our reasoning that Theorem~\ref{Theoremsufffid1} may not be a necessary condition). 
As for the former, we performed some checks by noting that every feasible point of the optimization we solve (namely,~\eqref{eq:approxfidoptSDP} with the constraint $\subn{p}^i \in  \breve{\mathcal{Q}}_{\mathcal{X},\mathcal{Y}} $ relaxed to the SDP hierarchy) gives us a feasible point of the original optimization~\eqref{eq:exactfidopt2} (albeit with the constraint $\vect{p}^{i} \in  \mathcal{Q}_{\mathcal{X},\mathcal{Y}}$ relaxed to the SDP hierarchy). We found that for points near the thresholds shown in Figs.~\ref{fig:3_2}--\ref{fig:2_2}, the corresponding feasible values in that original optimization were within $0.0003$ of the lower bounds we obtained, indicating that the bounds are almost tight. (For Fig.~\ref{fig:4_4} we found a bigger gap of about $0.03$ using a $6\times6$ lattice to find feasible points, but this is also not too large.)

Note that by applying Carath\'{e}odory's theorem for convex hulls, we can argue the minimum value in the optimization~\eqref{eq:approxfidopt1} can always be attained by a distribution $\operatorname{Pr}(i)$ with at most $d+2$ nonzero terms, where $d$ is the dimension of $\vect{p}$. This eventually implies that the minimum value in our final optimization~\eqref{eq:approxfidoptSDP} can be attained with at most $d+2$ of the subnormalized distributions $\subn{p}^i$ being nonzero. In practice, $d+2$ is often smaller than the number of affine bounds $h_i$ (i.e.~$|\mathcal{S}|$).
Hence if the optimization~\eqref{eq:approxfidoptSDP} is too large to solve directly, an alternative approach in principle is to run it for every subset of $\mathcal{S}$ with size $d+2$, then take the smallest of the resulting values. This reduces the size of each individual optimization, but comes at the cost of having to run many more of them. 

As another point regarding efficiency, note that our construction of $f_{n}$ involves a transformation to the facet representation of $\ptope_n$. While this can be quickly implemented when $\vect{x}$ is 2-dimensional, the transformation may be computationally demanding in high dimensions~\cite{henk1995basic}. It would be interesting to know whether there are more efficiently computable alternatives. 

A natural attempt would be to partition the domain of $f$ into triangles (more generally, simplices) and construct an affine lower bound in each triangle, yielding a piecewise affine lower bound on $f$. One benefit of this approach is that it may be usable (though not necessarily straightforward) in some cases when $f$ is not concave, whereas our current construction of $f_{n}$ relies heavily on concavity of $f$. However, it runs into the subtle issue that having $f_{n}$ be a pointwise minimum of affine functions is a stronger condition than simply requiring it to be piecewise affine; in particular, our analysis used the structure in the former. (Note that it is not useful to simply take the pointwise minimum of the affine bounds constructed this way --- there can be very large gaps between $f$ and the resulting bound.)

Still, it may be possible to adapt our analysis to this case. To sketch a rough outline, we would again aim to partition the sum in~\eqref{eq:approxfidopt1} into finitely many subsets, but this time by which of the triangles each $\vect{p}^i$ lies in. Transforming to a new feasible point as in~\eqref{eq:newQfeas} (here we would need to use the convexity of the triangles to argue that each $\tilde{\vect{p}}^j$ still lies within its defining triangle), we should be able to perform a similar analysis to reduce the objective to the form~\eqref{eq:newobjfct}, but with the summation index ranging over the triangles in the domain instead. However, to proceed further we would need to constrain each $\tilde{\vect{p}}^j$ to remain within the corresponding triangle, appearing as additional constraints in~\eqref{eq:approxfidopt1}. Since these constraints can be imposed as affine constraints, the final result should still be solvable using the SDP hierarchy. 

Finally, we note that in~\cite{himbeeck2019correlations}, the authors do not convert their optimization to the form~\eqref{eq:approxfidoptSDP}, but rather to a dual form via a somewhat different argument. A similar argument is possible in principle here (see \cite{HahnThesis}), but we choose to present our result in the form~\eqref{eq:approxfidoptSDP} since it seems most straightforward for implementation. 
We thank the authors of~\cite{himbeeck2019correlations} for clarifications on these different approaches. 

\subsection{Proof of Proposition~\ref{prop:neccondprop1}}
\label{sec:necproof}
\begin{proof}
After Alice and Bob conduct $n$ key-generating measurements, the resulting classical-classical-quantum tripartite state is of the form
\begin{equation}
\sum_{\mathbf{a},\mathbf{b} \in \{0,1\}^n} \operatorname{Pr}\left(\mathbf{a}\mathbf{b}\right)  \ketbra{\mathbf{a}\mathbf{b}}{\mathbf{a}\mathbf{b}} \otimes \rho_{\mathbf{E}|\mathbf{a}\mathbf{b}}  \ .
\end{equation}
Considering that Alice and Bob only take accepted blocks into account, i.e.~$D=1$, and Alice sends the message $\mathbf{M}=\mathbf{m}$, it is simple to construct the bipartite state $\rho_{C\mathbf{E}|\mathbf{M}=\mathbf{m} \land D=1 }$, which denotes the state that describes both the value of the bit $C$ and Eve's corresponding side-information. As $D=1$ implies that Alice's and Bob's measurement devices either output $\mathbf{m}$ or $\mathbf{\overline{m}}=\mathbf{m} \oplus \mathbf{1}$, the resulting state is given by
\begin{equation}
\rho_{C\mathbf{E}|\mathbf{M}=\mathbf{m} \land D=1 } = \sum_{c\in \{0,1\}} \frac{1}{2}\ketbra{c}{c} \otimes \omega_c \ ,
\end{equation}
where 
\begin{align} \label{condstate1}
\omega_0 &= \frac{\operatorname{Pr}\left(\mathbf{m}\mathbf{m}\right)\rho_{\mathbf{E}|\mathbf{m}\mathbf{m}} +\operatorname{Pr}\left(\mathbf{m}\mathbf{\overline{m}}\right)\rho_{\mathbf{E}|\mathbf{m}\mathbf{\overline{m}}} }{\operatorname{Pr}\left(\mathbf{m}\mathbf{m}\right)+\operatorname{Pr}\left(\mathbf{m}\mathbf{\overline{m}}\right)} \\
\label{condstate2}\omega_1 &= \frac{\operatorname{Pr}\left(\mathbf{\overline{m}}\mathbf{\overline{m}}\right)\rho_{\mathbf{E}|\mathbf{\overline{m}}\mathbf{\overline{m}}} +\operatorname{Pr}\left(\mathbf{\overline{m}}\mathbf{m}\right)\rho_{\mathbf{E}|\mathbf{\overline{m}}\mathbf{m}} }{\operatorname{Pr}\left(\mathbf{\overline{m}}\mathbf{\overline{m}}\right)+\operatorname{Pr}\left(\mathbf{\overline{m}}\mathbf{m} \right)}
\end{align}
denotes Eve's conditioned side-information \cite{tan2019advantage}. Moreover, after symmetrization, we get $\operatorname{Pr}\left(\mathbf{m}\mathbf{m}\right) = \operatorname{Pr}\left(\mathbf{\overline{m}}\mathbf{\overline{m}}\right)= \frac{\left(1-\epsilon\right)^n}{2^n}$ and $\operatorname{Pr}\left(\mathbf{m}\mathbf{\overline{m}}\right) = \operatorname{Pr}\left(\mathbf{\overline{m}}\mathbf{m}\right)= \frac{\epsilon^n}{2^n}$, which further simplifies \eqref{condstate1} and \eqref{condstate2}.

Eve's ability to correctly guess $C$ therefore depends on the distinguishability of $\omega_0$ and $\omega_1$. As $C$, and consequently $\omega_i$, is distributed uniformly, we may use the operational interpretation of the trace distance to derive Eve's optimal guessing probability. The optimal probability of guessing it incorrectly is thus given by 
\begin{equation} \label{Prob1appendix}
\operatorname{Pr}(C \neq C''|D=1) = 
\frac{1}{2} \left( 1-d \left(\omega_0, \omega_1\right) \right) \ .
\end{equation}
\begin{widetext}
We first consider Bob's guess. As $C \neq C'$ only if Alice measures $\mathbf{m}$ and Bob measures $\mathbf{\overline{m}}$ or vice versa,
\begin{align} \label{Prob2appendix}
\operatorname{Pr}(C \neq C'|D=1) =
\frac{\operatorname{Pr}\left(\mathbf{m}\mathbf{\overline{m}}\right) +\operatorname{Pr}\left(\mathbf{\overline{m}}\mathbf{m}\right)}{\operatorname{Pr}\left(\mathbf{m}\mathbf{m}\right)+\operatorname{Pr}\left(\mathbf{\overline{m}}\mathbf{\overline{m}}\right)+\operatorname{Pr}\left(\mathbf{m}\mathbf{\overline{m}}\right) +\operatorname{Pr}\left(\mathbf{\overline{m}}\mathbf{m}\right)} = \frac{\epsilon^n}{\epsilon^n+\left(1-\epsilon\right)^n} \eqqcolon \delta_n \ .
\end{align}

We now consider Eve's guess. By using the reverse triangle inequality of the 1-norm, we can get a lower bound on $d \left(\omega_0, \omega_1\right)$ in terms of $\delta_n$:
\begin{equation}
d \left(\omega_0, \omega_1\right) \geq \left(1-\delta_n \right) \cdot d \left(\rho_{\mathbf{E}|\mathbf{m}\mathbf{m}}, \rho_{\mathbf{E}|\mathbf{\overline{m}}\mathbf{\overline{m}}} \right) - \delta_n  \cdot d \left(\rho_{\mathbf{E}|\mathbf{m}\mathbf{\overline{m}}}, \rho_{\mathbf{E}|\mathbf{\overline{m}}\mathbf{m}} \right)  \ ,
\end{equation}
and substituting this into \eqref{Prob1appendix} yields
\begin{align}
\operatorname{Pr}(C \neq C''|D=1) &\leq \frac{1}{2} \left( 1-\left(1-\delta_n \right) \cdot d \left(\rho_{\mathbf{E}|\mathbf{m}\mathbf{m}}, \rho_{\mathbf{E}|\mathbf{\overline{m}}\mathbf{\overline{m}}} \right) + \delta_n  \cdot d \left(\rho_{\mathbf{E}|\mathbf{m}\mathbf{\overline{m}}}, \rho_{\mathbf{E}|\mathbf{\overline{m}}\mathbf{m}} \right)  \right) \nonumber\\
&= \frac{1}{2} \left( 1-\left(1-\delta_n \right) \cdot d \left(\rho_{\mathbf{E}|\mathbf{m}\mathbf{m}}, \rho_{\mathbf{E}|\mathbf{\overline{m}}\mathbf{\overline{m}}} \right)  \right) 
\ ,
\end{align}
where to get the second line we used the hypothesis $\rho_{E|01} = \rho_{E|10}$ (which implies $\rho_{\mathbf{E}|\mathbf{m}\mathbf{\overline{m}}} =  \rho_{\mathbf{E}|\mathbf{\overline{m}}\mathbf{m}} $).
The `Fuchs--van de Graaf-type' inequality~\eqref{eq:FvdGlow} then implies that
\begin{align}
\operatorname{Pr}(C \neq C''|D=1) \leq & \frac{1 }{2} \left(  1-\left(1-\delta_n \right) \cdot \left(1-\overlap \left(\rho_{\mathbf{E}|\mathbf{m}\mathbf{m}}, \rho_{\mathbf{E}|\mathbf{\overline{m}}\mathbf{\overline{m}}} \right) \right) \right) 
\ .
\end{align}
Moreover, note that we have 
$\overlap \left(\rho_{\mathbf{E}|\mathbf{m}\mathbf{m}}, \rho_{\mathbf{E}|\mathbf{\overline{m}}\mathbf{\overline{m}}} \right) = \overlap \left(\rho_{E|00}, \rho_{E|11} \right)^n$ (by applying the I.I.D.~assumption together with the multiplicative property~\eqref{eq:mult}, followed by the symmetry property~\eqref{eq:symm}).
Therefore a sufficient condition for \eqref{eq:nprobcond1} to hold is 
\begin{equation} \label{inequality1app}
\frac{1}{2} \left(1-\left(1-\delta_n \right) \cdot \left(1-\overlap \left(\rho_{E|00}, \rho_{E|11} \right)^n \right) \right) \leq \delta_n \ .
\end{equation}
We conclude the proof by showing that for all $n \in \mathbb{N}$, \eqref{eq:nnecfidcond2} is equivalent to \eqref{inequality1app}. Note that the inequality~\eqref{eq:FvdGlow} (together with the fact that $d(\rho,\sigma)\leq 1$) implies that $\overlap(\rho,\sigma)$ is always non-negative. Hence for all $n \in \mathbb{N}$ the inequality \eqref{eq:nnecfidcond2} is equivalent to 
\begin{equation} \label{nConjectureexp}
\overlap(\rho_{E|00}, \rho_{E|11})^n \leq \frac{\epsilon^n}{\left(1-\epsilon\right)^n} = \frac{\left(1-\epsilon\right)^n - \left(1-\epsilon\right)^n + \epsilon^n}{\left(1-\epsilon\right)^n} = 1 - \frac{\left(1-\epsilon\right)^n-\epsilon^n}{\left(1-\epsilon\right)^n}
\end{equation}
This inequality can be rewritten as
\begin{equation} \label{inequality2app}
1 - \overlap(\rho_{E|00}, \rho_{E|11})^n \geq \frac{\left(1-\epsilon\right)^n-\epsilon^n}{\left(1-\epsilon\right)^n} = \frac{1-2\delta_n}{1-\delta_n} \ .
\end{equation}
The previous inequality is equivalent to 
\begin{equation} \label{semifinalinequality}
1 -\left(1-\delta_n \right) \cdot \left(1-\overlap \left(\rho_{E|00}, \rho_{E|11} \right)^n \right) \leq 2 \delta_n 
\end{equation}
and dividing both sides by $2$ gives \eqref{inequality1app}.

\end{widetext}
\end{proof}

We note that the current gap between the sufficient and necessary conditions can be viewed as arising from the `Fuchs--van de Graaf-type' inequalities~\eqref{eq:FvdGlow} and \eqref{eq:FvdGupp}. This is because the sufficient condition (Theorem~\ref{Theoremsufffid1}) proof requires lower bounds on $\operatorname{H}(C|\mathbf{E} \mathbf{M}; D=1)$, whereas the conjectured necessary condition needs upper bounds. 
As noted in the supplemental material for~\cite{tan2019advantage}, the analysis we performed above also serves an alternative approach for proving Theorem~\ref{Theoremsufffid1} (the main proof in~\cite{tan2019advantage} instead used the inequality from~\cite{PhysRevLett.105.040505} to lower-bound ${\operatorname{H}(C|\mathbf{E} \mathbf{M}; D=1)}$). The main idea is that ${\operatorname{H}(C|\mathbf{E} \mathbf{M}; D=1)}$ can be lower bounded by the min-entropy, which simply equals ${-\log({(1-d \left(\omega_0, \omega_1\right))}/{2})}$. By performing an analysis similar to the above but using the inequality~\eqref{eq:FvdGupp} instead of~\eqref{eq:FvdGlow}, we end up (after some asymptotic analysis) with a sufficient condition for~\eqref{eq:entropycond1} to hold, which turns out to be exactly the same as Theorem~\ref{Theoremsufffid1} (except that since the only properties of $\overlap$ required for this argument are~\eqref{eq:symm}--\eqref{eq:FvdGupp}, it would hold with any $\overlap$ satisfying those properties in place of the fidelity $\operatorname{F}$). {Note that no assumption is needed on $d(\rho_{E|01}, \rho_{E|10})$ for this direction of the proof.} From this perspective, it appears that the main contribution to the gap is the difference between the bounds~\eqref{eq:FvdGlow} and~\eqref{eq:FvdGupp}, since other steps of the proof have comparatively small effects asymptotically.

However, regarding possible choices of distinguishability measure $\overlap$ in this generalized version of Theorem~\ref{Theoremsufffid1}, note that replacing the fidelity in the theorem statement with the pretty-good fidelity yields a worse result, due to the inequality $\operatorname{F}(\rho, \sigma) \geq \operatorname{F}_\mathrm{pg}(\rho, \sigma)$. (The opposite was true for the necessary condition, Proposition~\ref{prop:neccondprop1}.) The question remains of whether there are choices for the measure $\overlap$ that yield better bounds for the sufficient condition.

Finally, we remark that the above analysis essentially centers around distinguishing
$\rho_{\vect{E}|\vect{m}\vect{m}}$ and $\rho_{\vect{E}|\overline{\vect{m}}\overline{\vect{m}}}$.
Returning to our discussion of the quantum Chernoff bound, we observe that unless $\vect{M}= \vect{0}$ or $\vect{M}= \vect{1}$, these states are not of the form $\rho^{\otimes n}$ and $\sigma^{\otimes n}$ studied in the quantum Chernoff bound, though there are some structural similarities. {(If $\vect{M}$ were restrained to $\vect{M}= \vect{0}$ or $\vect{M}= \vect{1}$, one may consider only $\vect{M}= \vect{0}$, as $\vect{M}= \vect{1}$ can be thought of as a relabeling of measurement outcomes.)}
Since $\vect{M}$ is not restricted to these cases in general, we would need to study whether these other structurally similar states could still be analyzed using the proof techniques for the quantum Chernoff bound. 

As another perspective, note that the quantum Chernoff bound in fact satisfies almost all the properties~\eqref{eq:symm}--\eqref{eq:FvdGupp}~\cite{Aud12}. However, instead of the equality~\eqref{eq:mult}, it only satisfies the inequality $\overlap(\rho \otimes \rho',\sigma \otimes \sigma') \geq \overlap(\rho,\sigma) \overlap(\rho',\sigma')$. Looking through the proofs described above for the necessary versus sufficient conditions, this means that only the proof of the latter generalizes directly if we choose $\overlap=Q$. Unfortunately, Theorem~\ref{Theoremsufffid1} with the fidelity simply replaced by the quantum Chernoff bound is a worse result, because the quantities are related by ${\operatorname{F}(\rho, \sigma) \geq Q(\rho, \sigma)}$~\cite{Audenaert_2007} (similar to the previous situation for pretty-good fidelity). Hence an argument that simply follows the proof structure sketched above with $\overlap=Q$ would not yield a better result than the the original Theorem~\ref{Theoremsufffid1} statement based on fidelity --- to get better results using the quantum Chernoff bound, one would need a different proof structure.

\section*{Data Availability} 
The datasets produced in this work are available from the authors upon reasonable request.

\section*{Code Availability} 
The MATLAB code for this paper can be found at the following URL:
\begin{equation*}
\text{\href{https://github.com/Thomas0501/Fidelity-Optimization}{https://github.com/Thomas0501/Fidelity-Optimization}}
\end{equation*}

\section*{Acknowledgements}
We thank Thomas van Himbeeck for very helpful details regarding the SDP reduction in~\cite{himbeeck2019correlations}, and Renato Renner for feedback on this work. 
We also thank Raban Iten, Joseph M.~Renes, and Marco Tomamichel for useful discussions regarding the pretty-good fidelity and quantum Chernoff bound. 

This project was funded by the Swiss National Science Foundation via the National Center for Competence in Research for Quantum Science and Technology (QSIT), the Air Force Office of Scientific Research (AFOSR) via grant FA9550-19-1-0202, and the QuantERA project eDICT.

Part of this work was done while Thomas Hahn worked at the Weizmann Institute of Science under Rotem Arnon-Friedman. During his stay, he was supported by a research grant from the Marshall and Arlene Bennett Family Research Program.

The computations were performed with the NPAHierarchy function in QETLAB~\cite{qetlab}, using the CVX package~\cite{cvxpackage,cvxbook}, {as well the vert2con function created by Michael Kleder}.

\section*{Author Contributions} 
Both authors contributed to the theoretical analysis and algorithm development. T.~H.~implemented the algorithm and generated the corresponding data.

\section*{Competing Interests} 
The authors declare no competing interests.

\clearpage

\bibliography{Latex_Bib}
\bibliographystyle{ieeetr}

\end{document}